\documentclass[a4paper,10pt,twocolumn]{IEEEtran}
\usepackage{authblk}
\usepackage{amsmath}
\usepackage{graphicx}
\usepackage{color}
\usepackage{amssymb}
\usepackage{fancyhdr}
\usepackage{setspace}
\usepackage{wrapfig}
\usepackage{cite}
\usepackage{amsmath,amssymb}
\usepackage{cite}

\usepackage[bf,footnotesize]{caption}
\usepackage{makeidx}  
\usepackage{amsmath}
\usepackage{amssymb}
\usepackage{graphics}
\usepackage{psfrag}
\usepackage{epsfig}
\usepackage{subfigure}
\usepackage{array}
\usepackage{makeidx}  
\usepackage{amssymb}
\usepackage{bbm}
\usepackage{anysize}
\usepackage{mathrsfs}
\usepackage{dsfont}

\newtheorem{thm}{Theorem}

\newtheorem{remark}{Remark}

\newtheorem{lem}{Lemma}

\newtheorem{define}{Definition}

\allowdisplaybreaks[4]

\DeclareMathOperator{\ENC}{ENC}
\DeclareMathOperator{\DEC}{DEC}
\DeclareMathOperator*{\argmax}{arg\,max}
\DeclareMathOperator*{\argmin}{arg\,min}

\begin{document}
\IEEEoverridecommandlockouts
\title{Information Loss due to Finite Block Length in a Gaussian Line Network: An Improved Bound}
\author{\IEEEauthorblockN{Ramanan Subramanian, Badri N. Vellambi, and Ingmar Land}\\ \IEEEauthorblockA{Institute for Telecommunications Research, University of South Australia, Australia.\\ \{ramanan.subramanian, badri.vellambi, ingmar.land\}@unisa.edu.au}}
\maketitle
\thispagestyle{empty}
\pagestyle{empty}
\vspace{-0.50in}
\begin{abstract}
    A bound on the maximum information transmission rate through a cascade of Gaussian links is presented. The network model consists of a source node attempting to send a message drawn from a finite alphabet to a sink, through a cascade of Additive White Gaussian Noise links each having an input power constraint. Intermediate nodes are allowed to perform arbitrary encoding/decoding operations, but the block length and the encoding rate are fixed. The bound presented in this paper is fundamental and depends only on the design parameters namely, the network size, block length, transmission rate, and signal-to-noise ratio.
\end{abstract}
\section{Introduction}\label{sec:intro}
Transmission of messages through a series of links corrupted by noise is a situation that occurs frequently in communication networks. When the transmission block length is allowed to be arbitrarily large, it is quite simple to show (using the data-processing inequality) that the maximum information transfer rate is equal to the capacity of the weakest link. The possibilities in the finite block length regime are far less clear. Past work by Niesen \emph{et al.} in~\cite{Niesen07} and by us in~\cite{IZS12-Paper} have addressed this question for the Discrete Memoryless Channel (DMC) case and the Additive White Gaussian Noise (AWGN) case, respectively. These results are \emph{asymptotic} and provide scaling laws for the block length in terms of the number of nodes.

In this paper, we provide a universal \emph{non-asymptotic} bound on the maximum rate of information transfer for a line network consisting of a cascade of AWGN links. This complements and improves the asymptotic scaling results derived in~\cite{IZS12-Paper}. The bound derived here is universal in the following sense:
\begin{enumerate}
    \item While we assume that the block length and encoding rates are constant for all the nodes, we do not assume any particular structure for the channel codes and decision rules employed at any of the nodes.
    \item In addition, no assumption is made on the absolute/relative magnitudes of the network size and the block length.
\end{enumerate}
It is to be noted that the analysis in~\cite{IZS12-Paper} was found to be unsuitable to our requirement that the bound be non-asymptotic, and hence we take a totally new approach here.

The rest of the paper is organized as follows. In Section~\ref{sec:defs}, we introduce the notations and definitions used in the rest of the paper. In Section~\ref{sec:network-model}, we introduce the network and the signal transmission models. We then provide our main result followed by its derivation in Section~\ref{sec:analysis}, followed by a short discussion in Section~\ref{sec:discussion} that includes a comparison of our current results in relation to our previous results in~\cite{IZS12-Paper}.

\section{Notations and Definitions}\label{sec:defs}
Let $\mathds{R}$ be the set of all real numbers and $\mathds{N}$ be the set of all natural numbers. Natural logarithms are assumed unless the base is specified. The notation $\|\cdot\|$ represents $\mathcal{L}^2$ norm throughout. $\mathscr{S}_{M \times M}$ denotes the set of all $M \times M$ row-stochastic matrices, and $\mathscr{S}^*_{M \times M}$ denotes the set of all $M \times M$ row-stochastic matrices whose rows are identical.

Let  $N \in \mathds{N}$ denote the \emph{code length} or \emph{block length} of the transmission scheme. A \emph{code rate} $R > 0$ is a real number such that $2^{NR}$ is an integer. Let $\mathscr{M} \triangleq \{1,2,3,\ldots,2^{NR}\}$ be the \emph{message alphabet}.
\begin{define}
For a certain $P_0 \geq 0$, a rate $R$ length $N$ \emph{code} $\mathscr{C}$ with power constraint $P_0$ is an ordering of $M = 2^{NR}$ elements from $\mathds{R}^N$, called \emph{codewords}, such that the power of any codeword is lower than $P_0$:
\begin{equation}
    \mathscr{C} = \left(\mathbf{c}_1,\mathbf{c}_2,\mathbf{c}_3,\ldots,\mathbf{c}_{M}\right) \text{ s.t. } \forall {w \in \mathscr{M}}, \frac{1}{N}\|\mathbf{c}_w\|^2 \leq P_0.\nonumber
\end{equation}
\end{define}

\begin{define}
A rate $R$ length $N$ \emph{decision rule} $\mathscr{R} = \left(\mathcal{R}_1,\mathcal{R}_2,\ldots,\mathcal{R}_M\right)$ is an ordered partition of $\mathds{R}^N$ of size $M = 2^{NR}$.
\end{define}

\begin{define}
The \emph{encoding function} $\ENC_\mathscr{C}:\mathscr{M} \rightarrow \mathds{R}^N$ for a code $\mathscr{C}$ is defined by $\ENC_\mathscr{C}(w) = \mathbf{c}_w$, where $\mathbf{c}_w$ is the $w^{\text{th}}$ codeword in $\mathscr{C}$.
\end{define}

\begin{define}
The \emph{decoding function} $\DEC_\mathscr{R}:\mathds{R}^N \rightarrow \mathscr{M}$ for a decision rule $\mathscr{R}$ is defined by:
\begin{equation}
    \DEC_\mathscr{R}(\mathbf{y}) = w \text{ iff } \mathbf{y} \in \mathcal{R}_w,\nonumber
\end{equation}
where $\mathcal{R}_w$ is the $w^{\text{th}}$ partition in $\mathscr{R}$.

Let $\Omega_0 = \frac{ 2\pi^{\frac{N}{2}} }{\Gamma \left(\frac{N}{2}\right)}$, the solid angle of a $N$-sphere. Here, $\Gamma(\cdot)$ is the standard gamma function given by  $\Gamma(z) = \int_0^\infty  t^{z-1} e^{-t}\,{\rm d}t$.  We also define the following functions:
\begin{define}\label{def:functions}
    Let $Z_1,\ldots,Z_N$ be i.i.d. zero-mean unit-variance Gaussian random variables. Let for any $\gamma \geq 0$, $$\Phi_\gamma \triangleq \begin{cases}\cot^{-1} \left( \frac{ \sqrt{N\gamma} + z_1 }{ \sqrt{ \sum_{l=2}^{N} z_l^2 } } \right), & \sum_{l=2}^{N} z_l^2 > 0 \\ 0, \text{otherwise}.\end{cases}$$ Also, let $\forall v \in [0,\pi]$
    \[
        g(v) \triangleq \frac{ (N-1) \pi^{ \frac{N-1}{2} } }{ \Gamma\left( \frac{N+1}{2} \right) } \int_0^{v} \left( \sin \theta \right)^{N-2} \mathrm{d}\theta
    \]
    We then define the following function for $x \in \left[0,\Omega_0\right]$:
    \begin{eqnarray}
        \mathcal{Q}\left(x,N,\gamma\right) \triangleq \Pr\left[ g\left(\Phi_{\gamma}\right)\geq x \right].
    \end{eqnarray}
\end{define}
The above function is the same as $Q^*\left(\cdot\right)$ defined and used by Shannon in~\cite{Shannon-Papers}. In other words, computing $\mathcal{Q}\left(x,N,\gamma\right)$ gives the probability that a signal point on the power-constraint sphere $\| \mathbf{x} \|^2 = N P_0$ is displaced by a noise vector consisting of i.i.d. zero-mean unit-variance Gaussian random variables in each dimension outside an infinite right-circular cone of solid angle $x$ whose apex is at the origin and axis runs through the original signal point. Note that $\Phi_{\gamma}$ is a random variable whose probability distribution function has $N$ and $\gamma$ as parameters. The inverse cotangent function is assumed to have $[0, \pi]$ for its range so that it is continuous. Noting that $\cot^{-1}x$ will then be a decreasing function of $x$, we have the following remark about the monotonicity of the $\mathcal{Q}$-function w.r.t. $\gamma$:

\begin{remark}\label{rem:Q-is-monotonic}
    For any $\gamma_1, \gamma_2 > 0$ s.t. $\gamma_1 \geq \gamma_2$ and $x \in \left[0,\Omega_0\right]$,
    $$\mathcal{Q}\left(x,N,\gamma_2\right) \geq \mathcal{Q}\left(x,N,\gamma_1\right).$$
\end{remark}

The term $g \left( \Phi_{\gamma} \right)$ is equal to the solid angle of the cone formed by rotating the line joining the origin and the displaced signal point about the line joining the origin and the original signal point as the axis.
\end{define}
\section{Network Model}\label{sec:network-model}
The line network model to be considered is given in Fig.~\ref{fig:netwk-model}. There are $n+1$ nodes in the network identified by the indices $\left\{0,1,2,\ldots,n\right\}$. The $n$ hops in the network are each associated with noise variances ${\sigma_i}^2 \geq \sigma_0^2 > 0, 1 \leq i \leq n$. In other words, the noise variances can be different for each link, but they are equal to or greater than a certain minimum $\sigma_0^2$ that is strictly positive. Nodes $0, 1,\ldots,n-1$ choose codes $\mathscr{C}_0,\mathscr{C}_1,\ldots,\mathscr{C}_{n-1}$ respectively to transmit, and Nodes $1,2,\ldots,n$ choose decision rules $\mathscr{R}_1,\mathscr{R}_2,\ldots,\mathscr{R}_{n}$ for reception. From here on, for the sake of simplicity, we let $\ENC_i$ and $\DEC_i$ to denote $\ENC_{\mathscr{C}_i}$ and $\DEC_{\mathscr{R}_i}$, respectively. All the codes and decision rules have the same rate $R$ and block length $N$. Node 0 generates a random message $W \in \mathscr{M}$ with probability distribution $p_W(w)$ and intends to convey the same to Node $n$ through the noisy multihop path in the network. Each node estimates the message sent by the node in the previous hop from its noisy observation, encodes the message as a codeword, and transmits the resulting codeword to the next hop. The codeword transmitted by Node $i$, for any $0 \leq i \leq n-1$ is given by $\mathbf{X}_i = \ENC_i(\hat{W}_i)$, where $\hat{W}_i$ is the estimate of the message at Node $i$ after decoding (Note that $\hat{W}_0 = W$ in this notation). The observation received by Node $i$, for any $1 \leq i \leq n$ is given by $\mathbf{Y}_i$, which follows a conditional density function that depends on the codeword $\mathbf{X}_{i-1}$ sent by the previous node:
\begin{equation}
    p_{\mathbf{Y}_i|\mathbf{X}_{i-1}}\left(\mathbf{y}|\mathbf{x}\right) = \frac{1}{\left(2\pi{\sigma_i}^2\right)^\frac{N}{2}}e^{-\frac{\left\|\mathbf{y}-\mathbf{x}\right\|^2}{2\sigma_i^2}}.\label{eqn:conditional:distrib}
\end{equation}
The above density function follows from the assumptions of AWGN noise and memorylessness of the channel. The message $\hat{W}_i$ decoded by Node $i$ is given by $\hat{W}_i = \DEC_i(\mathbf{Y}_i)$. Note that the random variable $\hat{W}_n$ represents the message decoded by the final sink.
\begin{figure*}[htbp]
    \centering
    \psfrag{Node 0}{\scriptsize{Node 0}}
    \psfrag{Node i}{\scriptsize{Node $i$}}
    \psfrag{Node n}{\scriptsize{Node $n$}}
    \psfrag{pw}{\hspace{-0.075in}\tiny{$p_W$}}
    \psfrag{W}{\tiny{$W$}}
    \psfrag{ENC-0}{\hspace{-0.05in}\tiny{$\ENC_0$}}
    \psfrag{x0}{\tiny{$\mathbf{X}_0$}}
    \psfrag{p1yx}{\hspace{-0.1in}\tiny{$p_{\mathbf{Y}_1|\mathbf{X}_0}$}}
    \psfrag{y1}{\tiny{$\mathbf{Y}_1$}}
    \psfrag{piyx}{\hspace{-0.15in}\tiny{$p_{\mathbf{Y}_i|\mathbf{X}_{i-1}}$}}
    \psfrag{yi}{\hspace{-0.025in}\tiny{$\mathbf{Y}_i$}}
    \psfrag{DEC-i}{\hspace{-0.05in}\tiny{$\DEC_i$}}
    \psfrag{Wi}{\hspace{-0.05in}\tiny{$\hat{W}_i$}}
    \psfrag{ENC-i}{\hspace{-0.05in}\tiny{$\ENC_i$}}
    \psfrag{xi}{\tiny{$\mathbf{X}_i$}}
    \psfrag{pnyx}{\hspace{-0.2in}\tiny{$p_{\mathbf{Y}_n|\mathbf{X}_{n-1}}$}}
    \psfrag{yn}{\hspace{-0.05in}\tiny{$\mathbf{Y}_n$}}
    \psfrag{DEC-n}{\hspace{-0.05in}\tiny{$\DEC_n$}}
    \psfrag{W^}{\hspace{-0.05in}\tiny{$\hat{W}_n$}}
    \includegraphics[width=\textwidth]{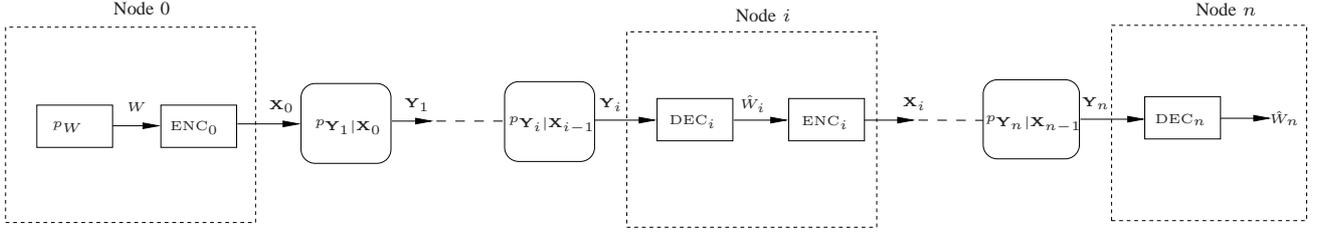}
    \caption{Line Network Model.}
    \label{fig:netwk-model}
\end{figure*}
\section{The main result and analysis}\label{sec:analysis}
The following theorem summarizes our main result.
\begin{thm}\label{thm:unif-full}
    In a line network employing any choice of rate $R$ length $N$ codes $\mathscr{C}_0,\mathscr{C}_1,\ldots,\mathscr{C}_{n-1}$ and rate $R$ dimension $N$ decision rules $\mathscr{R}_1,\mathscr{R}_2,\ldots,\mathscr{R}_n$,
    $$\mathcal{I}\left(W;\hat{W}_n\right) \leq NR \left[1 - M \mathcal{Q} \left( \frac{M-1}{M}\Omega_0, N, \frac{P_0}{\sigma_0^2} \right) \right]^n.$$
\end{thm}
We now delve into the proof of Theorem~\ref{thm:unif-full}. Let $p_{\hat{W}_i \mid \hat{W}_{i-1}}\left(k \mid j\right), \forall j,k \in \mathscr{M}$ denote the conditional probabilities induced by channel encoding, noisy reception, and decoding at the $i^{\text{th}}$ hop. For each hop $i$, let $\mathbf{P}_i$ be the $M \times M$ row-stochastic matrix whose entry in row $j$ and column $k$ is $p_{\hat{W}_i \mid \hat{W}_{i-1}}\left(k \mid j\right)$. Note that the $j^{\text{th}}$  row in $\mathbf{P}_i$ gives the conditional probability mass function on the estimate $\hat{W}_i$ of the original message $W$ at hop $i$, given that the message sent by Node $i - 1$ is $j$. Let
\[
    \mathbf{P} \triangleq \prod_{i = 1}^{n} \mathbf{P}_i.
\]
Then, $\mathbf{P}$ clearly represents the row-stochastic probability transition matrix between the original message $W$ and the message decoded at the sink $\hat{W}_n$. The transition matrix $\mathbf{P}$ along with $p_W$ (the probability mass function of the original message $W$) together induce a joint distribution between $W$ and $\hat{W}_n$. Our goal is to find an upper bound on $\mathcal{I}\left(W; \hat{W}_n\right)$, with the constraints given in Section~\ref{sec:network-model}.

For any $M \times M$ row-stochastic matrix $\mathbf{Q}$, define $\psi\left(\mathbf{Q} \mid p_W \right) \triangleq \mathcal{I}\left(W; \tilde{W}\right)$, where $\tilde{W}$ is a random variable conditionally dependent on $W$ according to the probability transition matrix $\mathbf{Q}$ and $W$ is drawn according to the distribution $p_W$ (which is the distribution of the message at Node 0). For simplicity, we just write $\psi\left(\mathbf{Q}\right)$ instead of $\psi\left(\mathbf{Q} \mid p_W \right)$ for the rest of the paper, assuming throughout that the specific distribution $p_W$ is used. Ultimately, our final bound is independent of $p_W$. We now have
\begin{eqnarray}
    \mathcal{I}\left(W;\hat{W}_n\right) &=& \psi \left( \prod_{i = 1}^n \mathbf{P}_i \right).\label{eqn:psi-prod}
\end{eqnarray}
Before proceeding further to bound $\mathcal{I}\left(W;\hat{W}_n\right)$, we introduce the following useful lemma:
\begin{lem}\label{lem:ss-stoch-prod}
    For any $\mathbf{Q}_1 \in \mathscr{S}^*_{M \times M}$ and any $\mathbf{Q}_2 \in \mathscr{S}_{M \times M}, \psi\left(\mathbf{Q}_1\mathbf{Q_2}\right) = 0$.
\end{lem}
\begin{proof}
    The result follows from noting that for any $\mathbf{Q} \in \mathscr{S}^*_{M \times M}$, $\psi\left(\mathbf{Q}\right) = 0$ and $\mathbf{Q}_1 \mathbf{Q}_2 \in \mathscr{S}^*_{M \times M}$ for $\mathbf{Q}_1 \in \mathscr{S}^*_{M \times M}$ and $\mathbf{Q}_2 \in \mathscr{S}_{M \times M}$.
\end{proof}

Now for each $i$, consider $\beta_i \in \left[0,1\right]$ such that
\begin{eqnarray}
    \mathbf{P}_i = \beta_i\mathbf{P}_{\beta_i} + \bar{\beta_i}\mathbf{P}_{\bar{\beta_i}},\label{eqn:stoch-matrix-convex-combo}
\end{eqnarray}
where $\bar{\beta_i} = 1 - \beta_i$, $\mathbf{P}_{\beta_i} \in \mathscr{S}_{M \times M}$ and $\mathbf{P}_{\bar{\beta_i}} \in \mathscr{S}^*_{M \times M}$. In other words for each $i$, $\mathbf{P}_i$ be expressed as a convex combination of two row-stochastic matrices, one of them being a steady-state matrix. From (\ref{eqn:psi-prod}),
\begin{align}
    \mathcal{I}\left(W;\hat{W}_n\right) &= \psi\Big(\prod_{i = 1}^n \mathbf{P}_i\Big)\hspace{-0.5mm}=\hspace{-0.5mm} \psi\Big(\hspace{-0.5mm} \big( \beta_1 \mathbf{P}_{\beta_1} \hspace{-0.5mm}+\hspace{-0.5mm} \bar{\beta_1}\mathbf{P}_{\bar{\beta_1}} \big)\hspace{-0.5mm} \prod_{i=2}^n\hspace{-0.5mm} \mathbf{P}_i \Big)\nonumber\\
    &\stackrel{(a)}{\leq} \beta_1 \psi \Big( \mathbf{P}_{\beta_1} \prod_{i=2}^n \mathbf{P}_i \Big) + \bar{\beta_1} \psi \Big( \mathbf{P}_{\bar{\beta_1}}  \prod_{i=2}^n \mathbf{P}_i \Big) \nonumber\\
    &\stackrel{(b)}{\leq}\beta_1 \psi \Big( \prod_{i=2}^n \mathbf{P}_i\Big),\nonumber
\end{align}
where (a) follows from the convexity property of mutual information w.r.t. the probability transition function, and (b) follows from applying the data processing inequality to the first term and Lemma~\ref{lem:ss-stoch-prod} to the second term. By induction, we have:
\begin{eqnarray}
    \mathcal{I}\left(W; \hat{W}_n\right) &\leq& \left( \prod_{i=1}^n \beta_i \right) \psi\left( I_{M} \right) = NR \prod_{i=1}^n \beta_i, \label{eqn:mutual-info-beta-product}
\end{eqnarray}
where $I_M$ is the $M \times M$ identity matrix. The above procedure is based on a key idea developed in the proof of Theorem V.1 in~\cite{Niesen07} in a different context. We have applied the same to facilitate a useful intermediate result given by (\ref{eqn:mutual-info-beta-product}). The remaining portion of the analysis that enables us to obtain the final bound involves novel steps.

We now need to determine how each $\mathbf{P}_i$ is to be split in the form given by (\ref{eqn:stoch-matrix-convex-combo}) in an optimal manner, to obtain the best possible bound using this approach. Specifically, we need $\beta_i$ to  be as small as possible for each $i$. Consider the following choice:
\begin{eqnarray}
    \beta_i &=& 1 - \sum_{k = 1}^M \min_j p_{\hat{W}_i \mid \hat{W}_{i-1}} \left( k \mid j \right)\nonumber\\
    P_{\bar{\beta_i}; j,k} &=& \frac{1}{1 - \beta_i} \min_{j'} p_{\hat{W}_i \mid \hat{W}_{i-1}} \left( k \mid j' \right),\label{eqn:convex-combo-params-choice}
\end{eqnarray}
where $P_{\bar{\beta_i}; j,k}$ denotes the element on $j^{\text{th}}$ row and $k^{\text{th}}$ column of the matrix $\mathbf{P}_{\bar{\beta_i}}$. Note that this matrix consists of identical rows, where each entry in any row is equal to the smallest element in the corresponding column of $\mathbf{P}_i$ scaled by a normalizing factor. The other matrix $\mathbf{P}_{\beta_i}$ is determined by substituting these expressions for $\beta_i$ and $\mathbf{P}_{\bar{\beta_i}}$ into (\ref{eqn:stoch-matrix-convex-combo}). Note that the two matrices $\mathbf{P}_{\beta_i}$ and $\mathbf{P}_{\bar{\beta}_i}$ determined thus will be stochastic for any $i$, and that $\beta_i \in [0,1]$. Hence, we can obtain $\mathbf{P}_i$ as a convex combination of two stochastic matrices in this manner for any $i$. The following lemma shows that the value of $\beta_i$ provided in (\ref{eqn:convex-combo-params-choice}) is the best possible value for the purpose of the bound in (\ref{eqn:mutual-info-beta-product}).

\begin{lem}\label{lem:best-beta}
    Let $\mathbf{Q} = \left[Q_{jk}\right] \in \mathscr{S}_{M \times M}$, $\mathbf{Q}_1 \in \mathscr{S}_{M \times M}, \mathbf{Q}_2 \in \mathscr{S}^*_{M \times M}$, and let $\beta \in [0,1]$ be chosen such that $\mathbf{Q} = \beta \mathbf{Q}_1 + (1-\beta)\mathbf{Q}_2$. Then, $\beta \geq 1 - \sum_{k = 1}^M \min_j Q_{jk}$.
\end{lem}
\begin{proof}
    Since $\mathbf{Q} = \beta \mathbf{Q}_1 + (1-\beta)\mathbf{Q}_2$, every element of the matrix $(1-\beta) \mathbf{Q}_2$ must be smaller than the corresponding element in $\mathbf{Q}$. Consider any column $k$ of $\mathbf{Q}_2$. All the elements in that column are equal to, say, $q_k$. It then follows that $(1-\beta)q_k \leq Q_{jk}$ for every $j$, and hence $(1-\beta)q_k \leq \min_j Q_{jk}$. Summing over all $k$ and noting that $\sum_{k=1}^M q_k = 1$, we obtain the desired result.
\end{proof}
%

Let $\mathscr{C}_{i-1}$ be the code used by Node $i - 1$ and let $\mathscr{R}_i$ be the decision rule used by Node $i$. As per the argument above, the optimal choice of $\beta_i$ for this link will be:
\begin{eqnarray}
    1 - \beta_i &=& \sum_{k = 1}^M \min_j p_{\hat{W}_i \mid \hat{W}_{i-1}} \left( k \mid j \right) \nonumber \\
    &=& \sum_{\mathcal{R} \in \mathscr{R}_i} \min_{\mathbf{c} \in \mathscr{C}_{i-1}} \int_{\mathcal{R}} \frac{e^{ -{\left\|\mathbf{y} - \mathbf{c}\right\|^2}/{2 \sigma_i^2 }}}{ \left( 2 \pi \sigma_i^2 \right)^{{N}/{2}} } \mathrm{d} \mathbf{y}.
\end{eqnarray}
In other words, we can write
\begin{eqnarray}
    \beta_i &=& 1 - \mu_{\sigma_i}\left(\mathscr{C}_{i-1},\mathscr{R}_i\right),
\end{eqnarray}
where for any $\sigma > 0$, rate $R$ length $N$ code $\mathscr{C}$, and rate $R$ dimension $N$ decision rule $\mathscr{R}$,
\begin{eqnarray}
    \mu_{\sigma}\left(\mathscr{C},\mathscr{R}\right) &\triangleq& \sum_{\mathcal{R} \in \mathscr{R}} \min_{\mathbf{c} \in \mathscr{C}} \int_{\mathcal{R}} \frac{e^{ -\left\|\mathbf{y} - \mathbf{c}\right\|^2 /2 \sigma^2 } }{ \left( 2 \pi \sigma^2 \right)^{N/2} } \mathrm{d}\mathbf{y}. \label{eqn:mu-defined}
\end{eqnarray}
We would like to find a lower bound on $\mu_{\sigma_i} \left(\mathscr{C}_{i-1}, \mathscr{R}_i\right)$ that depends only on the parameters $N, R, P_0$ and $\sigma_0$. To do so, we need the following three lemmas. Lemma~\ref{lem:opt-R-given-C} removes the dependency of the bound on the choice of the decision rule. Lemma~\ref{lem:ball-to-sphere} shows that we can restrict our choice of codes to those having all codewords that satisfy the power constraint with equality. For this class of codes, Lemma~\ref{lem:opt-code-config} gives a bound in the desired form, depending solely on $N, R, P_0$ and $\sigma_0$. From now on, we denote $(2 \pi \sigma_0^2)^{\frac{N}{2}}$ by $\eta$ for brevity.

\begin{lem}
    \label{lem:opt-R-given-C}
    Let $\mathscr{C} = \left( \mathbf{c}_1, \mathbf{c}_2, \ldots, \mathbf{c}_M \right)$ be a given rate $R$ length $N$ code. Further, let $\mathscr{R}^*\left(\mathscr{C}\right)$ be the decision rule given by $\left( \mathcal{R}^*_1,\mathcal{R}^*_2,\ldots,\mathcal{R}^*_M \right)$ where for $1 \leq i \leq M$,
    \begin{eqnarray}
        \mathcal{R}^*_i &=& \left\{ \mathbf{y} \in \mathds{R}^N \mid i = \argmax_{i'} \left\|\mathbf{y} - \mathbf{c}_{i'}\right\|\right\}.\nonumber
    \end{eqnarray}
    Then, for any rate $R$ dimension $N$ decision rule $\mathscr{R}$,
    \begin{eqnarray}
        \mu_{\sigma}\left(\mathscr{C},\mathscr{R}\right) \geq \mu_{\sigma}\left(\mathscr{C},\mathscr{R}^*\left(\mathscr{C}\right)\right).\nonumber
    \end{eqnarray}
\end{lem}
\begin{proof}
    For any code $\mathscr{C}$ and decision rule $\mathscr{R}$, we have:
    \begin{align}
    \eta \mu_{\sigma}\left(\mathscr{C},\mathscr{R}\right) &= \sum_{\mathcal{R} \in \mathscr{R}} \min_j \int_{\mathcal{R}} e^{ -\frac{ \left\|\mathbf{y} - \mathbf{c}_j\right\|^2 }{ 2 \sigma^2 } } \mathrm{d}\mathbf{y} \nonumber\\
    &\geq \sum_{\mathcal{R} \in \mathscr{R}} \int_{\mathcal{R}} \min_j  e^{ -\frac{ \left\|\mathbf{y} - \mathbf{c}_j\right\|^2 }{ 2 \sigma^2 } } \mathrm{d}\mathbf{y}\nonumber\\
    &= \int_{\mathds{R}^N} \min_j e^{ -\frac{ \left\|\mathbf{y} - \mathbf{c}_j\right\|^2 }{ 2 \sigma^2 } } \mathrm{d}\mathbf{y}\nonumber\\
    &= \sum_{k = 1}^{M}\int_{\mathcal{R}^*_k} \min_{j} e^{ -\frac{ \left\|\mathbf{y} - \mathbf{c}_j\right\|^2 }{ 2 \sigma^2 } } \mathrm{d}\mathbf{y}\nonumber\\
    &\stackrel{(a)}{=} \sum_{k = 1}^{M} \int_{\mathcal{R}^*_k} e^{ -\frac{ \left\|\mathbf{y} - \mathbf{c}_k\right\|^2 }{ 2 \sigma^2 } } \mathrm{d}\mathbf{y}\nonumber\\
    &\stackrel{(b)}{=} \sum_{k = 1}^{M} \min_{j} \int_{\mathcal{R}^*_k} e^{ -\frac{ \left\|\mathbf{y} - \mathbf{c}_j\right\|^2 }{ 2 \sigma^2 } } \mathrm{d}\mathbf{y}\nonumber\\
    &= \eta \mu_{\sigma} \left( \mathscr{C} , \mathscr{R}^* \left( \mathscr{C} \right) \right).
    \end{align}
    Here, (a) and (b) follow from the definition of $\mathscr{R}^* \left( \mathscr{C} \right)$: for any $\mathbf{y} \in \mathcal{R}^*_k, \argmin_{j} e^{ -\frac{ \left\|\mathbf{y} - \mathbf{c}_j\right\|^2 }{ 2 \sigma^2 } } = k$.
\end{proof}
It is useful to note that the decision rule $\mathscr{R}^* \left( \mathscr{C} \right)$ given by the above lemma is the same as the $(M - 1)^\text{th}$-order Voronoi partitioning (called ``farthest-point Voronoi partitioning'', see Section 3.3 in~\cite{Voronoi-Book}) of $\mathds{R}^N$ w.r.t. $\mathscr{C}$.

\begin{lem}
    \label{lem:ball-to-sphere}
    Let $\mathscr{C}$ be a code satisfying $\| \mathbf{c} \|^2 \leq N P_0, \forall \mathbf{c} \in \mathscr{C}$. Then, there exists a code $\mathscr{C}'$ such that $\forall \mathbf{c'} \in \mathscr{C}', \| \mathbf{c'} \|^2  = NP_0$ and $ \mu_{\sigma}\left(\mathscr{C},\mathscr{R}^*\left(\mathscr{C}\right)\right) \geq \mu_{\sigma}\left(\mathscr{C}',\mathscr{R}^*\left(\mathscr{C}'\right)\right) $.
\end{lem}
\begin{proof}
    Let $\mathscr{C} = \left( \mathbf{c}_1, \ldots, \mathbf{c}_M \right)$ and let $\mathscr{R}^*\left( \mathscr{C} \right) = \left( \mathcal{R}^*_1,\ldots,\mathcal{R}^*_M\right)$. Consider a codeword that lies strictly inside the ball $\| \mathbf{x} \|^2 < NP_0$. If no such codeword exists, the statement of the lemma is trivially true with $\mathscr{C}' = \mathscr{C}$. For the non-trivial case, we can assume that such a codeword exists. Let $\mathbf{c}_{k_0}$ be that codeword. Consider the decision region $\mathcal{R}^*_{k_0} = \left\{ \mathbf{y} \in \mathds{R}^N \mid k_0 = \argmax_{i'} \left\|\mathbf{y} - \mathbf{c}_{i'}\right\|\right\} \in \mathscr{R}^*\left( \mathscr{C} \right)$. The region $\mathcal{R}^*_{k_0}$ (if non-empty) is convex since $\mathscr{R}^*\left( \mathscr{C} \right)$ is a Voronoi tesselation of order $M-1$ and since Voronoi cells of any order are convex regions (see Property OK.1 in Section 3.2 of~\cite{Voronoi-Book}). Hence, there exists a unique point $\mathbf{z}_{k_0}$ in $\mathcal{R}^*_{k_0}$ nearest to $\mathbf{c}_{k_0}$. By moving the codeword at $\mathbf{c}_{k_0}$ along the line joining $\mathbf{c}_{k_0}$ and $\mathbf{z}_{k_0}$ away from the latter, the distance from the codeword to every point in $\mathcal{R}^*_{k_0}$ is increased. We continue thus until the codeword is moved to the surface of the power-constraint sphere, at say $\mathbf{c'}_{k_0}$. Let us call the resulting code $\mathscr{C}_1$. Note that $\mathscr{C} \backslash \left\{\mathbf{c}_{k_0}\right\} = \mathscr{C} \backslash \left\{\mathbf{c'}_{k_0}\right\}$. Now consider
    \begin{align}
        \eta \mu_{\sigma}\left(\mathscr{C},\mathscr{R}^*\left(\mathscr{C}\right)\right) &= \sum_{k = 1}^{M} \min_{j} \int_{\mathcal{R}^*_k} {e^{ -\frac{ \left\|\mathbf{y} - \mathbf{c}_j\right\|^2 }{ 2 \sigma^2 } }}\mathrm{d}\mathbf{y}\nonumber\\
        &\stackrel{(c)}{=} \sum_{k = 1}^{M}\int_{\mathcal{R}^*_k} {e^{ -\frac{ \left\|\mathbf{y} - \mathbf{c}_k\right\|^2 }{ 2 \sigma^2 } } } \mathrm{d}\mathbf{y}\nonumber\\
        \begin{split}&= \sum_{\substack{k = 1 \\ k \neq k_0}}^{M}\int_{\mathcal{R}^*_k} {e^{ -\frac{ \left\|\mathbf{y} - \mathbf{c}_k\right\|^2 }{ 2 \sigma^2 } }} \mathrm{d}\mathbf{y} \\ &\quad \quad + \int_{\mathcal{R}^*_{k_0}} {e^{ -\frac{ \left\|\mathbf{y} - \mathbf{c}_{k_0}\right\|^2 }{ 2 \sigma^2 } }} \mathrm{d}\mathbf{y}\end{split}\nonumber\\
        \begin{split}& \stackrel{(d)}{\geq}  \sum_{\substack{k = 1 \\ k \neq k_0}}^{M}\int_{\mathcal{R}^*_k} {e^{ -\frac{ \left\|\mathbf{y} - \mathbf{c}_k\right\|^2 }{ 2 \sigma^2 } }} \mathrm{d}\mathbf{y}\\ &\quad \quad + \int_{\mathcal{R}^*_{k_0}} {e^{ -\frac{ \left\|\mathbf{y} - \mathbf{c'}_{k_0}\right\|^2 }{ 2 \sigma^2 } }} \mathrm{d}\mathbf{y}\end{split}\nonumber\\
        \begin{split}&\geq \sum_{\substack{k = 1 \\ k \neq k_0}}^{M} \min_{\mathbf{c} \in \mathscr{C}_1} \int_{\mathcal{R}^*_k} {e^{ -\frac{ \left\|\mathbf{y} - \mathbf{c}\right\|^2 }{ 2 \sigma^2 } }} \mathrm{d}\mathbf{y}\\ & \quad \quad + \min_{\mathbf{c} \in \mathscr{C}_1} \int_{\mathcal{R}^*_{k_0}} {e^{ -\frac{ \left\|\mathbf{y} - \mathbf{c}\right\|^2 }{ 2 \sigma^2 } }}\mathrm{d}\mathbf{y}\end{split}\nonumber\\
        &= \eta \mu_{\sigma}\left(\mathscr{C}_1,\mathscr{R}^*\left(\mathscr{C}\right)\right)\nonumber\\
        &\stackrel{(e)}{\geq} \eta \mu_{\sigma}\left(\mathscr{C}_1,\mathscr{R}^*\left(\mathscr{C}_1\right)\right).
    \end{align}
    Here, (c) follows from (b) in the proof of Lemma~\ref{lem:opt-R-given-C}, (d) follows from the construction of $\mathbf{c'}_{k_0}$ so that for every $\mathbf{y} \in \mathcal{R}^*_{k_0}, \| \mathbf{y} - \mathbf{c}_{k_0} \| \leq \| \mathbf{y} - \mathbf{c'}_{k_0} \|$, and (e) follows from Lemma~\ref{lem:opt-R-given-C}. Note also that we have only treated the case where $\mathcal{R}^*_{k_0}$ is non-empty. If on the other hand, that decision region was empty, we can move the codeword at $\mathbf{c}_{k_0}$ along any arbitrary direction. For such a case inequality (b) becomes an equality since the integrals over $\mathcal{R}^*_{k_0}$ would be zero.
    From a given code $\mathscr{C}$, we can thus obtain a code $\mathscr{C}_1$ having one more codeword on the surface of the power-constraint sphere, also satisfying $\mu_{\sigma}\left(\mathscr{C},\mathscr{R}^*\left(\mathscr{C}\right)\right) \geq \mu_{\sigma}\left(\mathscr{C}_1,\mathscr{R}^*\left(\mathscr{C}_1\right)\right)$. We can repeat this process several times to eventually obtain a code $\mathscr{C}'$ with all codewords on the power constraint sphere.
\end{proof}

\begin{lem}
    \label{lem:opt-code-config}
    For any rate $R$ length $N$ code $\mathscr{C}$ satisfying the power constraint $\forall \mathbf{c} \in \mathscr{C}, \| \mathbf{c} \|^2 = N P_0$,
    $$ \mu_{\sigma}\left(\mathscr{C},\mathscr{R}^*\left(\mathscr{C}\right)\right) \geq \mathcal{Q}\left( \frac{M-1}{M} \Omega_0, N, \frac{P_0}{\sigma^2} \right).$$
\end{lem}
\begin{proof}
    For any code $\mathscr{C}$ that satisfies the requirements of the lemma, the decision regions in $\mathscr{R}^*\left( \mathscr{C} \right)$ consists of pyramids with their apex at the origin and extending out to infinity (see Appendix~\ref{appdx:voronoi-sphere-pyramids} for a proof). Assume that each of these regions $\left\{\mathcal{R}^*_k\right\}_{k=1}^M$ cut out a surface of area $\Omega_k$ on the unit $N$-sphere centered at the origin. Note that for each codeword $\mathbf{c}_k \in \mathscr{C}$ corresponding to message $k$, the decision region $\mathcal{R}^*_k$ contains the point $-\mathbf{c}_k$. Consider any term in the summation of the expression for $\mu_{\sigma} \left( \mathscr{C}, \mathscr{R}^*\left( \mathscr{C} \right) \right)$:
    \begin{eqnarray}
         \min_{j} \int_{\mathcal{R}^*_k} \frac{e^{ -{ \left\|\mathbf{y} - \mathbf{c}_j\right\|^2 }/{ 2 \sigma^2 } } }{ \left( 2 \pi \sigma^2 \right)^{\frac{N}{2}} }\mathrm{d}\mathbf{y} &=& \int_{\mathcal{R}^*_k} \frac{e^{ -{ \left\|\mathbf{y} - \mathbf{c}_k\right\|^2 }/{ 2 \sigma^2 } } }{ \left( 2 \pi \sigma^2 \right)^{\frac{N}{2}} }\mathrm{d}\mathbf{y}.\nonumber
    \end{eqnarray}
    The right hand side of the above equation is equal to the probability of the event $E_1$ that the transmitted codeword in $\mathds{R}^N$ located at $\mathbf{c}_k$ on the sphere $ \| \mathbf{x} \|^2 = N P_0$ is displaced by the noise vector into a specific region $\mathcal{R}^*_k$ that contains the point $-\mathbf{c}_k$. Now consider the probability of the event $E_2$ that the same transmitted codeword is displaced into the $N$-dimensional circular cone $\mathcal{C}^*_k$ that has its apex at the origin, axis running through $-\mathbf{c}_k$, and cutting out a surface of area $\Omega_k$ on the unit sphere centered at the origin (i.e., the solid angle of the $N$-dimensional circular cone is $\Omega_k$). We claim that the probability of $E_1$ cannot be smaller than the probability of $E_2$. A proof of this claim is provided in Appendix~\ref{appdx:pyramid-to-cone}. The probability of the event $E_2$ is equal to $\mathcal{Q} \left( \Omega_0-\Omega_k, N, P_0/\sigma^2 \right)$, from the definition of the $\mathcal{Q}$-function in Definition~\ref{def:functions}. Hence,
    \begin{eqnarray}
        \mu_{\sigma}\left(\mathscr{C},\mathscr{R}^*\left(\mathscr{C}\right)\right) &=& \sum_{k = 1}^M \min_{j} \int_{\mathcal{R}^*_k} \frac{e^{{ \left\|\mathbf{y} - \mathbf{c}_j\right\|^2 }/{ 2 \sigma^2 } } }{ \left( 2 \pi \sigma^2 \right)^{N/2} }\mathrm{d}\mathbf{y} \nonumber\\
        &\geq& \sum_{k = 1}^M \mathcal{Q} \left( \Omega_0-\Omega_k, N, \frac{P_0}{\sigma^2} \right).\label{eqn:q-before-jensen}
    \end{eqnarray}
    Noting that $\mathcal{Q}$ is a convex function of $\Omega_0 - \Omega_k$ (See Section III in~\cite{Shannon-Papers}), we apply Jensen's inequality to (\ref{eqn:q-before-jensen}):
    \begin{eqnarray}
        \mu_{\sigma}\left(\mathscr{C},\mathscr{R}^*\left(\mathscr{C}\right)\right) &\geq& M \frac{1}{M} \sum_{k = 1}^M \mathcal{Q} \left( \Omega_0-\Omega_k, N, \frac{P_0}{\sigma^2} \right)\nonumber\\
        &\geq& M \mathcal{Q} \left( \frac{ \sum_{k = 1}^M \left( \Omega_0-\Omega_k \right) }{M}, N, \frac{P_0}{\sigma^2} \right) \nonumber\\
        &=& M \mathcal{Q} \left( \frac{M-1}{M}\Omega_0, N, \frac{P_0}{\sigma^2}\right),
    \end{eqnarray}
    since $\sum_{k=1}^M \Omega_k = \Omega_0$.
\end{proof}

We are now ready to prove Theorem~\ref{thm:unif-full}.
\begin{proof}[Proof of Theorem~\ref{thm:unif-full}]
    Consider any link $i$. For the code $\mathscr{C}_{i-1}$ satisfying $\| \mathbf{c} \|^2 \leq N P_0, \forall \mathbf{c} \in \mathscr{C}_{i-1}$, we can apply Lemma~\ref{lem:ball-to-sphere} to construct another code $\mathscr{C}'_{i-1}$ such that $\forall \mathbf{c} \in \mathscr{C}'_{i-1}, \| \mathbf{c} \|^2 = N P_0$ and $ \mu_{\sigma_i}\left(\mathscr{C}_{i-1},\mathscr{R}^*\left(\mathscr{C}_{i-1}\right)\right) \geq \mu_{\sigma_i}\left(\mathscr{C}'_{i-1},\mathscr{R}^*\left(\mathscr{C}'_{i-1}\right)\right) $.  We then have:
    \begin{eqnarray}
        \mu_{\sigma_i}\left(\mathscr{C}_{i-1},\mathscr{R}_i\right) &\stackrel{(a)}{\geq}& \mu_{\sigma_i}\left(\mathscr{C}_{i-1},\mathscr{R}^*\left(\mathscr{C}_{i-1}\right)\right)\nonumber\\
        &\stackrel{(b)}{\geq}& \mu_{\sigma_i}\left(\mathscr{C}'_{i-1},\mathscr{R}^*\left(\mathscr{C}'_{i-1}\right)\right)\nonumber\\
        &\stackrel{(c)}{\geq}& M \mathcal{Q} \left( \frac{M-1}{M}\Omega_0, N, \frac{P_0}{\sigma_i^2}\right)\nonumber\\
        &\stackrel{(d)}{\geq}& M \mathcal{Q} \left( \frac{M-1}{M}\Omega_0, N, \frac{P_0}{\sigma_0^2}\right).\label{eqn:link-i:mu-bound-Q}
    \end{eqnarray}
    In the above chain of equations (a), (b), and (c) follow from Lemma~\ref{lem:opt-R-given-C}, Lemma~\ref{lem:ball-to-sphere} (as discussed above), and Lemma~\ref{lem:opt-code-config} respectively. Inequality (d) follows from Remark~\ref{rem:Q-is-monotonic}, since $\sigma_i^2 \geq \sigma_0^2$. Recalling that $\beta_i = 1 - \mu_{\sigma_i}\left(\mathscr{C}_{i-1},\mathscr{R}_i\right)$ and applying (\ref{eqn:link-i:mu-bound-Q}) to (\ref{eqn:mutual-info-beta-product}), we have the desired result:
    \[
        \mathcal{I}\left(W;\hat{W}_n\right) \leq NR \left[1 - M \mathcal{Q} \left( \frac{M-1}{M}\Omega_0, N, \frac{P_0}{\sigma_0^2} \right) \right]^n.\nonumber
    \]
\end{proof}
\section{Discussion}\label{sec:discussion}
We had mentioned in Section~\ref{sec:intro} that the bound presented in the current paper improves and complements the bound provided by~\cite{IZS12-Paper}. In this section, we demonstrate this fact with a comparison plot. The bound given by~\cite{IZS12-Paper} is:
\begin{eqnarray}
    \mathcal{I}\left( W; \hat{W}_n \right) &\leq& 2^{NR} \left( 1 - e^{-N E(P_0 / \sigma_0^2) } \right) ^{n \left(1-\epsilon\right)},\label{eqn:expo-old}
\end{eqnarray}
where for any $S \geq 0$,
\begin{eqnarray}
    E\left(S\right) &\triangleq& \frac{(S+2) + \sqrt{(S+2)^2 - 4}}{4}\nonumber\\
                    & & + \frac{1}{2}\log\left\{ (S+2)+\sqrt{(S+2)^2-4} \right\}.\nonumber
\end{eqnarray}
for asymptotically large $n$. The above bound decays with $n$ as $\left( 1 - e^{-N E(P_0 / \sigma_0^2) } \right)$ for any code rate $R$, while the bound given by Theorem~\ref{thm:unif-full} decays as $\left( 1 - M\mathcal{Q}\left(\frac{M-1}{M}\Omega_0, N, \frac{P_0}{\sigma_0^2} \right) \right)$. Now, let
\[
    E_{as}\left(R,S\right) \triangleq - \lim_{N \rightarrow \infty} \frac{\log\left(M\mathcal{Q}\left(\frac{M-1}{M}\Omega_0, N, S\right)\right)} {N}.
\]
We now investigate how these two bounds compare when $N$ is asymptotically large, by comparing the values of $E(S)$ and $E_{as}\left(R,S\right)$ where $S = P_0/\sigma_0^2$. To do so we obtained $E_{as}\left(R,S\right)$ as a function of $R$ and $S$ using the asymptotic analysis in~\cite{Shannon-Papers}. The expression for $E_{as}\left(R,S\right)$ is given below, with a justification in Appendix~\ref{appdx:asymptotics}:
$$E_{as}\left(R,S\right) = S - \frac{1}{2}\sqrt{S} G  \cos \theta - \log\left(G \sin \theta \right),$$ where $$G = \frac{1}{2} \left( \sqrt{S} \cos \theta + \sqrt{ 4 + S \cos^2 \theta }\right),$$ and $\theta = \pi - \sin^{-1} 2^{-R}$. Shown in Fig.~\ref{fig:plot-exponent-compare} is a plot of $E_{as}(R,S)/E(S)$ as a function of $S$, repeated for various $R$. As can be seen from the plot $E_{as}(R,S)$ is always smaller than $E(S)$ for any $R$ and $S$, thus showing that the bound obtained in Theorem~\ref{thm:unif-full} is tighter than the one given by (\ref{eqn:expo-old}). The current bound is also seen to be better when the SNR $S$ is not very high.
\begin{figure}[htbp]
    \centering
    \psfrag{x-axis}{\hspace{-0.375in}\scriptsize{Signal-to-Noise Ratio, $S$ in dB}}
    \psfrag{y-axis}{\hspace{-0.25in}\scriptsize{$E_{as}(R,S)/E(S)$}}
    \includegraphics[width = \columnwidth]{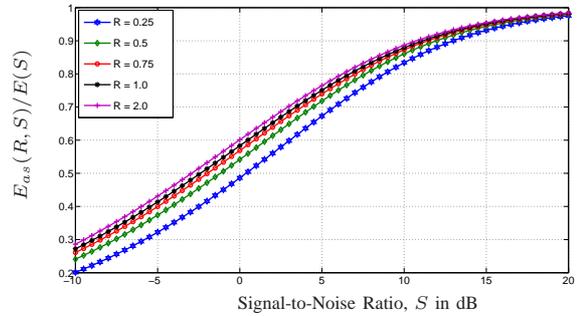}
    \caption{Comparison of exponents for large $N$.}
    \label{fig:plot-exponent-compare}
\end{figure}
\thispagestyle{empty}
\bibliographystyle{IEEEbib}
\bibliography{References}

\newpage
\appendices
\section{The Farthest-point Voronoi Tessellation for Points on a Sphere} \label{appdx:voronoi-sphere-pyramids}
\begin{lem}\label{lem:voronoi-sphere-pyramids}
    Given any $M, N \in \mathds{N}$, the non-empty cells in the farthest-point Voronoi tessellation of $\mathds{R}^N$ w.r.t. any set of $M$ points on an $N$-sphere of radius $A > 0$ are all semi-infinite pyramids.
\end{lem}
\begin{proof}
    Consider the farthest-point Voronoi tessellation of $\mathds{R}^N$ w.r.t. a code $\mathscr{C} = \left( \mathbf{c}_1,\ldots,\mathbf{c}_M \right) \text{s.t.} \|\mathbf{c}_i\|^2 = A^2, 1 \leq i \leq M$. Consider any $\mathbf{x} \in \mathds{R}^N$ and assume without loss of generality that it is contained in $\mathcal{R}^*_1$, the farthest-point Voronoi cell for $\mathbf{c}'_1$. In that case, we have for all $i$ s.t. $ 1 \leq i \leq M$,
    \begin{eqnarray}
        \|\mathbf{c}_1 - \mathbf{x}\|^2 &\geq& \|\mathbf{c}_i - \mathbf{x}\|^2 \nonumber\\
        \Rightarrow \| \mathbf{c}_1 \|^2 + \| \mathbf{x} \|^2 - 2 \langle \mathbf{c}_1,\mathbf{x} \rangle  &\geq& \| \mathbf{c}_i \|^2 + \| \mathbf{x} \|^2 - 2 \langle \mathbf{c}_i,\mathbf{x} \rangle \nonumber\\
        \Rightarrow - \langle \mathbf{c}_1,\mathbf{x} \rangle  &\geq& - \langle \mathbf{c}_i,\mathbf{x} \rangle.\label{eqn:pyramid-proof-inner-product}
    \end{eqnarray}
    Consider any $\alpha \geq 0$. We claim that $\alpha \mathbf{x}$ is also contained in $\mathcal{R}^*_1$. This can be shown to be true by applying (\ref{eqn:pyramid-proof-inner-product}) to the expansion of $\| \mathbf{c}_1 - \alpha \mathbf{x}\|^2$:
    \begin{eqnarray}
        \| \mathbf{c}_1 - \alpha \mathbf{x}\|^2 &=& \| \mathbf{c}_1 \|^2 + \alpha^2 \| \mathbf{x} \|^2 - 2 \alpha \langle \mathbf{c}_1,\mathbf{x} \rangle \nonumber\\
        &\geq& \| \mathbf{c}_i \|^2 + \alpha^2 \| \mathbf{x} \|^2 - 2 \alpha \langle \mathbf{c}_i,\mathbf{x} \rangle \nonumber\\
        &=& \| \mathbf{c}_i - \alpha \mathbf{x}\|^2,\nonumber
    \end{eqnarray}
    for all $i$ s.t. $1 \leq i \leq M$. Hence, $\alpha \mathbf{x}$ is also contained in $\mathcal{R}^*_1$. Generalizing this, we have shown that any non-empty Voronoi cell that contains a point $\mathbf{x}$ also contains the point $\alpha \mathbf{x}$ for any $\alpha \geq 0$. Such a region is a semi-infinite pyramid by definition.
\end{proof}

\section{Proof of the claim in Lemma~\ref{lem:opt-code-config}}\label{appdx:pyramid-to-cone}
Consider the $N-1$ dimensional cross-section of the pyramid $\mathcal{R}^*_k$ cut out by a sphere of radius $R$ centered at the origin. This will be an arbitrary spherical polygon. The cross-section of the cone $\mathcal{C}^*_k$ by the same sphere will be a spherical cap with its center at $-\mathbf{c}_k$. The axis of the cone cuts through the spherical cap at its center. The non-overlapping regions of such a spherical cap and a polygon are illustrated in Fig.~\ref{fig:pyramid-to-cone}. Since the both the cross sections have the same surface area $R^N \Omega_k$, the surface areas of the non-overlapping parts of both the cross-sections (indicated as $A_1$ and $A_2$ and by two different shadings in Fig.~\ref{fig:pyramid-to-cone}) are equal. Now, every point in the shaded region $A_1$ on the polygon is nearer to $\mathbf{c}_k$ than any point in $A_2$ is to $\mathbf{c}_k$. This is because the former lies outside the spherical cap centered at $-\mathbf{c}_k$ while the latter is inside the same. This in turn implies that the angle $\theta_2 \in [0, \pi]$ between the axis of the cone and the line joining any point $\mathbf{y}_2$ on $A_2$ and the origin is smaller than the angle $\theta_1 \in [0, \pi]$ between the axis and the line joining any point $\mathbf{y}_1$ on $A_1$ and the origin, as shown in the right hand side of Fig.~\ref{fig:pyramid-to-cone}. This in turn implies that $\mathbf{y}_1$ is closer to $\mathbf{c}_k$ than $\mathbf{y}_2$ is to $\mathbf{c}_k$, as shown below:
\begin{eqnarray}
    \| \mathbf{c}_k - \mathbf{y}_1 \|^2 &=& N P_0 + R^2 + 2 R \sqrt{N P_0} \cos{\theta_1}\nonumber\\
    & \leq & N P_0 + R^2 + 2 R \sqrt{N P_0} \cos{\theta_2} \nonumber\\
    &=& \| \mathbf{c}_k - \mathbf{y}_2 \|^2.\nonumber
\end{eqnarray}
This in turn means that the integral of the density function of the Gaussian noise vector with center at $\mathbf{c}_k$ over the volume $\mathcal{R}^*_k$ is greater than the integral over the volume $\mathcal{C}^*_k$. The former is the probability of the event $E_1$ and the latter is the probability of the event $E_2$.
\begin{figure*}[htbp]
    \centering
    \psfrag{A}{\scriptsize{Axis}}
    \psfrag{A1}{\scriptsize{$A_1$}}
    \psfrag{A2}{\scriptsize{$A_2$}}
    \psfrag{Y1}{\scriptsize{$\mathbf{y}_2$}}
    \psfrag{Y2}{\scriptsize{$\mathbf{y}_1$}}
    \psfrag{O}{\scriptsize{$O$}}
    \psfrag{t1}{\scriptsize{$\theta_2$}}
    \psfrag{t2}{\scriptsize{$\theta_1$}}
    \psfrag{r}{\hspace{-0.10in}\scriptsize{$\sqrt{N P_0}$}}
    \psfrag{R}{\scriptsize{$R$}}
    \psfrag{X}{\scriptsize{$\mathbf{c}_k$}}
    \includegraphics[width=0.66\textwidth]{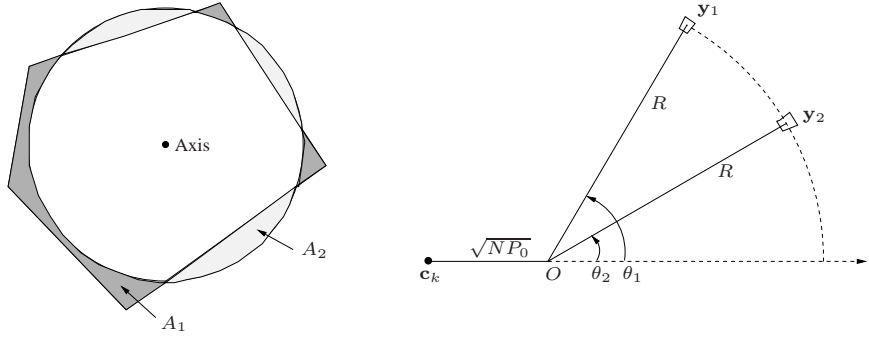}
    \caption{Illustration of the fact that the non-overlapping portion of the pyramidal cross-section is nearer to the original codeword than the non-overlapping portion of the conical cross section.}
    \label{fig:pyramid-to-cone}
\end{figure*}

\section{Asymptotic exponential decay of the $\mathcal{Q}$ function with $N$}\label{appdx:asymptotics}
 Though it is hard to express $\mathcal{Q}$ in terms of elementary functions, it is easy to obtain asymptotic approximations when the block length $N$ is very large. The idea is to use Shannon's computation of the sphere-packing exponent. Shannon derives a bound on $\mathcal{Q}\left(.\right)$ as a function of the cone angle $\theta \in [0,\pi]$ instead of the solid angle $\Omega$ since this makes asymptotic analysis easier. This results in a bound for $\mathcal{Q}\left(.\right)$ that decays exponentially in $N$, with the exponent being $$E_L\left(\theta\right) = \frac{P_0}{2\sigma_0^2} - \frac{1}{2}\sqrt{ \frac{P_0}{\sigma_0^2} } G  \cos \theta - \log\left(G \sin \theta \right),$$ where $$G = \frac{1}{2} \left( \sqrt{ \frac{P_0}{\sigma_0^2} } \cos \theta + \sqrt{ 4 + \frac{P_0}{\sigma_0^2} \cos^2 \theta }\right).$$ The bound on $\mathcal{Q}\left(.\right)$ for a given $\Omega$ can then be evaluated numerically or by any other means, since there is a one-to-one correspondence between the cone angle $\theta_0$ and the solid angle $\Omega$ (see Fig.~\ref{fig:cone-solid}): $$\Omega\left(\theta_0\right) = \frac{ (N-1) \pi^{ \frac{N-1}{2} } }{ \Gamma\left( \frac{N+1}{2} \right) } \int_0^{\theta_0} \left( \sin \theta_0 \right)^{N-2} \mathrm{d} \theta_0.$$ The particular case of interest in~\cite{Shannon-Papers} is $\mathcal{Q} \left( \frac{1}{M}\Omega_0, N, \frac{P_0}{\sigma_0^2} \right)$, which corresponds to the cone angle $\theta = \sin^{-1} 2^{-R}$  and the sphere-packing lower bound is obtained thus (see pages 620 and 625 in~\cite{Shannon-Papers}).  Our bound involves $\mathcal{Q} \left( \frac{M-1}{M}\Omega_0, N, \frac{P_0}{\sigma_0^2} \right)$ instead, and hence we will have to evaluate the exponent $E_L\left(\theta\right)$ with $\theta = \pi - \sin^{-1} 2^{-R}$ instead, giving us the following result:
\begin{figure}[htbp]
    \centering
    \psfrag{Omega}{\scriptsize{$\Omega \left(\theta_0\right)$}}
    \psfrag{t}{\scriptsize{$\theta_0$}}
    \includegraphics[width=0.40\columnwidth]{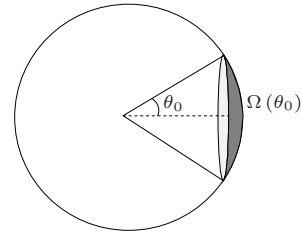}
    \caption{Relation between solid angle and cone angle.}
    \label{fig:cone-solid}
\end{figure}
\begin{eqnarray}
    \mathcal{I}\left(W;\hat{W}_n\right) \lessapprox \left( 1 - e^{- N E_{as}\left(R,P_0/\sigma_0^2\right)} \right)^n,
\end{eqnarray}
where $E_{as}\left(R,S\right)$ is as shown in (\ref{eqn:E_as-appdx}).
\begin{figure*}[htb]
\centering
\begin{eqnarray}
    E_{as}\left(R,S\right) &=& \frac{S}{2^{2R + 2}} \left( (2^{2R} + 1) + (2^{2R} - 1) \sqrt{ 1 + \frac{2^{2R+2}}{\left(2^{2R} - 1\right) S} }\right)\nonumber\\
    & & \ \ \ \ \ + \frac{1}{2} \log \left[ 2^{2R} + \frac{S}{2} (2^{2R} - 1) \left( \sqrt{ 1 + \frac{2^{2R+2}}{\left(2^{2R} - 1\right) S} } + 1\right)  \right] - R \log 2.\label{eqn:E_as-appdx}
\end{eqnarray}
\end{figure*}
\end{document}